\documentclass[11pt]{article}
\usepackage{graphicx}
\usepackage{subfigure}
\usepackage{float}
\usepackage{amsfonts,amsbsy,amssymb,amsmath,amsthm,amsfonts,mathtools}
\usepackage{verbatim}
\usepackage{parskip}
\usepackage[T1]{fontenc}
\usepackage{multicol}
\usepackage[onehalfspacing]{setspace}
\usepackage{color}
\usepackage[autostyle]{csquotes}
\usepackage{pgf}
\usepackage{hyperref}
\usepackage{thmtools}
\usepackage{caption}
\usepackage{soul}
\usepackage{pdflscape}
\usepackage{relsize}
 
\declaretheoremstyle[%
spaceabove=-6pt,%
spacebelow=6pt,%
headfont=\normalfont\itshape,%
postheadspace=1em,%
qed=\qedsymbol%
]{mystyle} 

\hypersetup{
	colorlinks=true,
	linkcolor=blue,
	filecolor=blue,      
	urlcolor=blue,
	citecolor=blue,
	pdftitle={Overleaf Example},
	pdfpagemode=FullScreen,
}

\newtheorem{thm}{Theorem}[section]
\newtheorem{prop}{Proposition}[section]
\newtheorem{cor}[thm]{Corollary}
\newtheorem{defn}[thm]{Definition}

\newtheorem{lemma}[thm]{Lemma}

\newtheorem{remark}[prop]{Remark}
\newtheorem{example}[thm]{Example}

\newcommand{\FF}{{\mathbb F}}

\newcommand{\CC}{\mathcal{C}}
\newcommand{\CD}{\mathcal{D}}

\newcommand{\cc}{\mathbf{c}}
\newcommand{\ba}{\mathbf{a}}
\newcommand{\bx}{\mathbf{x}}
\newcommand{\by}{\mathbf{y}}
\newcommand{\bb}{\mathbf{b}}

\newcommand{\rk}{\mathrm{rk}}
\newcommand{\G}{\mathcal{G}}
\newcommand{\HH}{\mathcal{H}}

\title{On $\ell$-rank additive intersection pairs (RAIP) of codes}

\author{Sanjit Bhowmick$^{1}$, Kuntal Deka$^{1}$ and Sihem Mesnager$^{2}$
	\footnote{
		$^{1}$Department of Electrical and Electronics Engineering, Indian Institute of Technology Guwahati,\\
		Assam, 781039, India. Email:sanjitbhowmick@rnd.iitg.ac.in ;kuntaldeka@iitg.ac.in\\
		$^{2}$Department of Mathematics, University of Paris VIII, F-93526 Saint-Denis, Laboratory Analysis, Geometry and Applications, LAGA, University Sorbonne Paris Nord, CNRS, UMR 7539, F-93430, Villetaneuse, France and Telecom Paris, Polytechnic institute of Paris, 91120 Palaiseau, France. Email: smesnager@univ-paris8.fr 
	}
}

\begin{document}
	\maketitle
	\begin{abstract} 
		
		This paper introduces and studies \(\ell\)-rank additive intersection pairs (RAIP) of codes over finite fields for a given positive integer \(\ell\). The notion of \(\ell\)-RAIP provides a common framework that generalizes additive complementary dual (ACD) codes, additive complementary pairs (ACP) of codes, and the hull of an additive code. We establish necessary and sufficient conditions characterizing when a pair of additive codes forms an \(\ell\)-RAIP. Furthermore, for (q>2), we prove that any pair of additive codes is monomially equivalent to an ACP of codes. As a consequence, for (q>3), every additive code is monomially equivalent to an ACD code. { A key contribution of the paper is a general construction method for \(\ell\)-RAIP of codes derived from self-orthogonal additive codes, which yields families of \((\ell+1)\)-RAIP of codes under suitable conditions.} In addition, several explicit constructions of \(\ell\)-RAIP of codes are presented.
		
	\end{abstract}
	
	\noindent\textbf{Keywords:} Additive code, Character of a finite abelian group, Linear $\ell$-intersection pairs of codes, $\ell$-rank additive intersection pairs of codes,  Hull codes,  Generalised Reed Solomon (GRS) code.
	
	\noindent\textbf{2020 AMS Classification Code:} 94B05; 94B60.

	\section{Introduction}\label{sec:intr}
	Additive codes are defined over additive groups, with additive codes over finite fields being constructed using the additive group of a finite field. Although all linear codes over finite fields are additive codes, {the converse does not valid,} meaning that the class of additive codes is broader than that of linear codes. Additive codes over finite fields have gained significant attention due to their relevance in quantum error correction and quantum computing \cite{Cal98,Ezerman}. Additionally, additive codes find applications in secret sharing schemes (see \cite{Kim17}). { These codes can be studied under various dualities arising from the characters of finite abelian groups} while common dualities such as Euclidean, Hermitian, and Trace dualities serve as important examples, other dualities can also be employed depending on the requirements of a particular application. Thus, additive codes provide a flexible framework in which different dualities can be utilized to achieve specific coding objectives. Delsarte introduced additive codes in 1973 using association schemes \cite{Del73}. For both symmetric and non-symmetric dualities, an additive code satisfies the cardinality condition, which states that the product of the number of elements in the code and the number of elements in its dual, with respect to any duality, equals the total number of elements in the ambient space (see \cite{Woo99}). Determining the optimal parameters for additive quaternary codes was first proposed by Blokhuis and Brouwer \cite{Blokhuis04}.
	
	The development of additive coding theory has progressed through several important contributions. Linear codes have been further generalized to additive codes over finite fields. It is important to note that additive codes are non-linear codes that possess elegant algebraic structures and can be fully described in terms of a basis. The concept of an $\FF_2$-linear additive code over $\FF_4$ was first introduced by Calderbank et al. \cite{Cal98}. Their work also established a connection between additive codes and quantum stabilizer codes through the construction of quantum stabilizer codes from $\FF_2$-linear additive self-orthogonal codes over $\FF_4$. Subsequently, Bierbrauer et al. \cite{Bierbrauer} extended the theory of additive codes to arbitrary finite fields. In related work, Huffman \cite{Huff13} investigated additive codes over finite fields and their dual codes with respect to both Euclidean and Hermitian trace bilinear forms. He also derived the MacWilliams identity and a Singleton-type bound for additive codes over finite fields. { Recently, Shi et al. \cite{shi22} pointed out that, in the design of orthogonal direct-sum masking schemes used to protect cryptographic implementations, only the additivity and complementarity properties of a code are required, rather than its linearity. Such masking schemes are employed to protect against side-channel attacks (SCAs), which exploit information leaked through physical implementations, and fault injection attacks (FIAs), which intentionally induce faults to reveal secret information. This observation broadens the range of codes that can be used in secure cryptographic applications. As a result, additive codes with complementary duals (ACD codes) can also be utilized as countermeasures against both passive and active side-channel analyses in embedded cryptosystems.} The Hamming distance of the corresponding ACD code determines the security level of such schemes. Motivated by this observation, Shi et al. \cite{shi222} investigated $\FF_2$-linear additive codes with complementary duals over $\FF_4$ and derived ACD codes over $\FF_4$ from binary codes using ordinary and Hermitian trace bilinear forms.

	Recently, Dougherty et al. generalized the concepts of Linear Complementary Dual (LCD) codes and one-dimensional hull codes to develop additive complementary dual (ACD) codes \cite{Dougherty22} and one-rank hull codes \cite{Dou124}. They focused on quaternary ACD and one-rank hull codes and determined optimal parameters for these codes over symmetric dualities for lengths \(n \leq 10\). This work motivated Agrawal and Sharma to introduce a specific class of non-symmetric dualities called skew-symmetric dualities and to study ACD codes under these conditions (see \cite{Agrawal24, Agrawa24}). While the notion of complementary duality focuses on codes whose intersection with their dual is trivial, a natural generalization is to consider pairs of codes whose intersection has a prescribed dimension. This perspective led to the introduction of linear \(\ell\)-intersection pairs of codes by Guenda et al. \cite{Guenda2019}. This concept serves as a generalization of Linear Complementary Dual (LCD) codes and linear complementary pairs (LCPs) of codes, as discussed in \cite{BDM23, CG18, CMT18, CMTQ19, Mas92}. Two linear codes over finite fields are said to form a linear \(\ell\)-intersection pair if the dimension of their intersection is \(\ell\). The authors provided a characterization of such pairs in terms of the corresponding generator and parity-check matrices of the codes. They also demonstrated that linear \(\ell\)-intersection pairs of codes can be used to construct entanglement-assisted quantum error-correcting codes.
	
	{ Motivated by these developments in both additive coding theory and \(\ell\)-intersection pairs of codes, the present paper extends the notion of linear \(\ell\)-intersection pairs to the additive setting by introducing \(\ell\)-rank additive intersection pairs (RAIP) of codes. The paper also provides two necessary and sufficient conditions for two additive codes, \(\CC\) and \(\CD\), over finite fields to form \(\ell\)-rank additive intersection pairs (RAIP).}

	The present paper is organized as follows. In Section~\ref{sec:2}, we we first review the necessary background on linear codes, particularly linear \(\ell\)-intersection pairs of codes, and then present relevant material on additive codes, characters of finite groups, and other concepts required in our study. In Section~\ref{sec:3}, we provide a characterization of \(\ell\)-rank additive intersection pairs (RAIP) of codes using parity-check and generator matrices. Furthermore, we show that any additive code over \(\FF_{q^m}\), with \(q \geq 4\), is equivalent to an ACD code. In Sections~\ref{sec:4} and \ref{sec:5}, we present several constructions of \(\ell\)-rank additive intersection pairs (RAIP) of codes. Finally, Section~\ref{sec:conclusion} concludes the paper.
	
	\section{Some preliminaries}\label{sec:2} 
	Throughout this paper, the notation $|E|$ denotes the cardinality of a finite set $E$, and $\ell$ is a positive integer. This section provides an overview of the foundational definitions and pivotal results necessary for deriving our main algebraic contributions.
	
	Let $q=p^e$ be a prime power, and let $\mathbb{F}_q$ represent the finite field of order $q$. For a positive integer $n$, the set $\mathbb{F}_q^n$ of all $n$-tuples over $\mathbb{F}_q$ forms an $n$-dimensional vector space. A linear code $C$ of length $n$ and dimension $k$ over $\mathbb{F}_q$ is a $k$-dimensional subspace of $\mathbb{F}_q^n$, whose elements are called codewords. The dual code of $C$ is defined as $C^\perp=\{\mathbf{x}\in \mathbb{F}_q^n \mid \langle \mathbf{x}, \mathbf{c} \rangle=0 \text{ for all } \mathbf{c}\in C\}$. A linear code $C$ is a linear complementary dual (LCD) code if $C\cap C^\perp=\{\mathbf{0}\}$. Furthermore, a pair of codes $\{C,D\}$ of the same length $n$ is called a linear complementary pair (LCP) if $\dim C+\dim D=n$ and $C\cap D=\{\mathbf{0}\}$. More generally, for a non-negative integer $\ell$, the pair $\{C,D\}$ is an $\ell$-dimensional linear intersection pair ($\ell$-DLIP) if $\dim(C\cap D)=\ell$. The following proposition provides a necessary and sufficient condition for a pair of linear codes to form an $\ell$-DLIP.
	
	\begin{prop}\cite[Lemma~4, Remark~1, 2]{Bhowmick23} \label{p-2.1}
		For $i=1,2$, let $C_i$ be a linear $[n, k_i]$ code over $\mathbb{F}_{q^m}$ with a generator matrix $\mathcal{G}_i$ and parity-check matrix $\mathcal{H}_i$. The following statements are equivalent
		\begin{enumerate}
			\item[a)] The pair $\{C_1, C_2\}$ is an $\ell$-DLIP of codes;
			\item[b)] $\mathrm{rk}\left(\begin{array}{cc}
				\mathcal{G}_1 \\
				\mathcal{G}_2 
			\end{array}\right)=k_1+k_2-\ell$; 
			\item[c)] $\mathrm{rk}\left(\begin{array}{cc}
				\mathcal{H}_1 \\
				\mathcal{H}_2 
			\end{array}\right)=2n-k_1-k_2-\ell$.
		\end{enumerate}
	\end{prop}
	{ Linear codes over finite fields admit a natural generalization known as \emph{additive codes}. Throughout this manuscript, let $m \ge 2$ be an integer and let $\FF_{q^m}$ denote the finite field with $q^m$ elements.
		
		Since $\FF_{q^m}$ is an $m$-dimensional vector space over its subfield $\FF_q$, the ambient space $\FF_{q^m}^n$ can naturally be viewed as an $(mn)$-dimensional vector space over $\FF_q$, with addition and scalar multiplication performed component-wise. An additive code of length $n$ over $\FF_{q^m}$ is simply an $\FF_q$-linear subspace of $\FF_{q^m}^n$. Thus, additive codes need not be linear over the larger field $\FF_{q^m}$, but these codes are linear over the base field $\FF_q$.
		If an additive code $\mathcal C$ contains $q^k$ codewords, where $0 \le k \le mn$, then $\mathcal C$ is called an $(n,q^k)$ additive code. Equivalently, $\mathcal C$ has dimension $k$ as an $\FF_q$-vector space, and we write
		$\mathrm{rk}_q(\mathcal C)=k.$
		A generator matrix of $\mathcal C$ is any matrix over $\FF_{q^m}$ whose rows form an $\FF_q$-basis of $\mathcal C$.}
	
	\noindent
	Further, for a finite abelian group $G$, an additive code $\mathcal{C}$ of length $n$ over $G$ can be viewed as an additive subgroup of $G^n$. To define a generalized duality over such groups, we utilize the framework of group characters.
	
	\noindent
	A character of $G$ is a group homomorphism $\chi: G \rightarrow \mathbb{C}^*$, where $\mathbb{C}^*$ denotes the multiplicative group of nonzero complex numbers. The set of all characters of $G$ forms a group $\widehat{G}$ under pointwise multiplication. It is a standard result that $\widehat{G}$ is abstractly isomorphic to $G$, though not canonically. Let $\varphi: G \rightarrow \widehat{G}$ be a fixed choice of isomorphism. For every element $u \in G$, we denote its corresponding character as $\chi_u = \varphi(u)$. 
	
	\noindent
	To evaluate how elements interact under this choice of isomorphism, we construct a character table or \textit{duality matrix} $M$. Let the elements of $G$ be indexed as $\{g_1, g_2, \ldots, g_{|G|}\}$. The duality matrix $M$ is a square matrix of size $|G| \times |G|$ whose rows are indexed by the characters $\chi_{g_i}$ and whose columns are indexed by the group elements $g_j$. The entry in the $i$-th row and $j$-th column is defined by evaluating the character on the group element
	$$M_{i,j} = \chi_{g_i}(g_j).$$
	
	{ \begin{example}
			Consider the finite field $\mathbb{F}_4 = \{0, 1, \alpha, \alpha^2\}$, where $\alpha^2 + \alpha + 1 = 0$. As an additive group, $(\mathbb{F}_4, +)$ is isomorphic to the Klein four-group $\mathbb{Z}_2 \times \mathbb{Z}_2$. Let $\text{Tr}: \mathbb{F}_4 \rightarrow \mathbb{F}_2$ be the standard field trace function. For each $u \in \mathbb{F}_4$, we define the character $\chi_u(v) = (-1)^{\text{Tr}(uv)}$. The resulting duality matrix $M$ over $\mathbb{F}_4$ is explicitly given by
			$$
			M = \begin{array}{r|cccc}
				& 0 & 1 & \alpha & \alpha^2 \\
				\hline
				\chi_{0} & 1 & 1 & 1 & 1 \\
				\chi_{1} & 1 & 1 & -1 & -1 \\
				\chi_{\alpha} & 1 & -1 & -1 & 1 \\
				\chi_{\alpha^2} & 1 & -1 & 1 & -1 
			\end{array}
			$$
	\end{example}}
	\noindent
	For two vectors $\mathbf{u} = (u_1, u_2, \ldots, u_n)$ and $\mathbf{v} = (v_1, v_2, \ldots, v_n)$ in $G^n$, their inner product is defined as
	$$\langle \mathbf{u}, \mathbf{v}\rangle_{M} = \prod\limits_{i=1}^{n}\chi_{u_i}(v_i).$$
	Because this inner product maps into the multiplicative group $\mathbb{C}^*$, the identity element is $1$. Thus, we say that $\mathbf{u}$ and $\mathbf{v}$ are orthogonal with respect to $M$ if and only if $\langle \mathbf{u}, \mathbf{v} \rangle_M = 1$. 
	
	\noindent
	Consequently, the orthogonal dual code depends fundamentally on our choice of $M$. We denote this generalized dual code as $\mathcal{C}^M$ (to distinguish it from the standard Euclidean dual $\mathcal{C}^\perp$) and define it as
	$$\mathcal{C}^M = \{\mathbf{u} \in G^n \mid \langle \mathbf{u}, \mathbf{v} \rangle_M = 1, \text{ for all } \mathbf{v} \in \mathcal{C}\}.$$
	Note that $|\mathcal{C}||\mathcal{C}^M|=|G|^n$ \cite{Woo99}. For a generator matrix $\G$ and parity check matrix $\HH$ of $\CC$, one can easily check that
	\begin{equation}\label{eq-GP}
		\G\odot_{M}\HH^\top=\textbf{1},   
	\end{equation}
	where $\mathbf{1}$ represents the all-ones matrix.
	Additionally, we say an additive code $\CC$ over a finite group $G$ is an additive complementary dual (ACD) with respect to $M$ if it $ \CC \cap \CC^M = \{\textbf 0\}$. Further, $\CC$ is self-orthogonal if $\CC\subseteq \CC^M$. If a pair $\{\CC, \CD\}$ satisfies $|\CC||\CD|=|G^n|$ and $\CC\cap\CD=\{\textbf{0}\}$, we call the pair an additive complementary pair (ACP) of codes.
	
	\section{Characterization of $\ell$-RAIP of codes}\label{sec:3}
	\begin{defn}
		For a non-negative integer $\ell$, a pair $\{\mathcal{C}, \mathcal{D}\}$ of additive codes of length $n$ over $\mathbb{F}_{q^m}$ is called an $\ell$-rank additive intersection pair ($\ell$-RAIP) of codes if $\mathrm{rk}_q(\mathcal{C} \cap \mathcal{D}) = \ell$.
	\end{defn}
	
	The structural properties of $\ell$-RAIPs generalize several classical additive code configurations. We collect these connections informally below.
	{
		\begin{remark}\label{rem-facts}
			From Definition 3.1, the following structural relations follow immediately.
			\begin{enumerate}
				\item[\textit{(i)}] An additive code $\CC$ is an ACD if $\{\CC, \CC^\perp\}$ is $0$-RAIP (for details information, see
				\cite{Agrawal24,Dougherty22,choi23}).
				\item[\textit{(ii)}] An additive $0$-RAIP $\{\CC, \CD\}$ with $\mathrm{rk}_q(\CC)+\mathrm{rk}_q(\CD)=mn$ is an ACP (for details
				information, see \cite{BD24}).
				\item[\textit{(iii)}] The rank of hull of any additive code $\CC$ is $\ell$ if and only if $\{\CC, \CC^\perp\}$ is $\ell$-RAIP (for details information about, see \cite{Agrawa24,Dou124}).
			\end{enumerate}
		\end{remark}
	}
	
	We now characterize these additive intersection pairs using matrix representations. { To connect character theory with linear algebra, we introduce an entry-wise matrix logarithm. Let $\xi \in \mathbb{C}^{\ast}$ be a primitive $p$-th root of unity, and let $A=[A_{i,j}]$ be a matrix whose entries belong to the cyclic group $\langle \xi \rangle$. We define the matrix logarithm of $A$ with respect to $\xi$, denoted by $\log_{\xi}(A)$, as the matrix obtained by applying the discrete logarithm base $\xi$ to each entry of (A); that is,
		\[
		\log_{\xi}(A)=\bigl[\log_{\xi}(A_{i,j})\bigr].
		\]
		Since every entry $A_{i,j}$ is a power of $\xi$, the resulting entries of $\log_{\xi}(A)$ lie in $\mathbb{F}_{p}$.
	}
	
	\begin{thm}\label{th-3.1}
		For $i=1,2$, let $\mathcal{C}_i$ be an additive $(n, q^k)$ code over $\mathbb{F}_{q^m}$ with generator matrix $\mathcal{G}_i$ and parity-check matrix $\mathcal{H}_i$. The pair $\{\mathcal{C}_1, \mathcal{C}_2^M\}$ is an $\ell$-RAIP (i.e., $\mathrm{rk}_q\left(\CC_1 \cap \CC_2^M\right)=\ell$) if and only if
		$$ \mathrm{rk}_q\left(\log_{\xi}\left(\mathcal{G}_1 \odot_{M} \mathcal{G}_2^\top \right)\right) = \mathrm{rk}_q\left(\log_{\xi}\left(\mathcal{G}_2 \odot_{M} \mathcal{G}_1^\top\right)\right) = k - \ell $$
		and 
		$$ \mathrm{rk}_q\left(\log_{\xi}\left(\mathcal{H}_1 \odot_{M} \mathcal{H}_2^\top\right)\right) = \mathrm{rk}_q\left(\log_{\xi}\left(\mathcal{H}_2 \odot_{M} \mathcal{H}_1^\top\right)\right) = nm - k - \ell. $$
	\end{thm}
	{

		\begin{proof} Let the rows of $\mathcal{G}_1$ be denoted by \[ \{\mathbf{g}_{11},\mathbf{g}_{12},\ldots,\mathbf{g}_{1k}\}, \] which form an $\mathbb{F}_q$-basis of $\mathcal{C}_1$. Likewise, let \[ \{\mathbf{g}_{21},\mathbf{g}_{22},\ldots,\mathbf{g}_{2k}\} \] be the rows of $\mathcal{G}_2$, forming an $\mathbb{F}_q$-basis of $\mathcal{C}_2$. Consider an arbitrary codeword $\mathbf{x}\in \mathcal{C}_1\cap \mathcal{C}_2^M$. Since $\mathbf{x}\in \mathcal{C}_1$, it can be uniquely written as \[ \mathbf{x}=\sum_{i=1}^{k} n_i\mathbf{g}_{1i}, \qquad n_i\in\mathbb{F}_q. \] On the other hand, because $\mathbf{x}\in \mathcal{C}_2^M$, it is orthogonal to every basis vector of $\mathcal{C}_2$ with respect to the $M$-inner product. Hence, \[ \langle \mathbf{x},\mathbf{g}_{2j}\rangle_M = \prod_{i=1}^{k} \left( \chi_{\mathbf{g}_{1i}}(\mathbf{g}_{2j}) \right)^{n_i} = 1, \qquad j=1,2,\ldots,k. \] Since all character values belong to the cyclic group generated by the primitive $p$-th root of unity $\xi$, we may write \[ \chi_{\mathbf{g}_{1i}}(\mathbf{g}_{2j}) = \xi^{s_{ij}}, \qquad s_{ij}\in\mathbb{F}_q. \] Substituting this expression into the above relation yields \[ \prod_{i=1}^{k} \left(\xi^{s_{ij}}\right)^{n_i} = \xi^{\sum_{i=1}^{k}s_{ij}n_i} = 1, \] which implies \[ \sum_{i=1}^{k}s_{ij}n_i\equiv 0 \pmod p. \] Consequently, the coefficients $(n_1,\ldots,n_k)$ satisfy the homogeneous linear system \[ \begin{pmatrix} s_{11} & s_{21} & \cdots & s_{k1}\\ s_{12} & s_{22} & \cdots & s_{k2}\\ \vdots & \vdots & \ddots & \vdots\\ s_{1k} & s_{2k} & \cdots & s_{kk} \end{pmatrix} \begin{pmatrix} n_1\\ n_2\\ \vdots\\ n_k \end{pmatrix} = \begin{pmatrix} 0\\ 0\\ \vdots\\ 0 \end{pmatrix}. \] By construction, the coefficient matrix is precisely the transpose of the entrywise logarithm of the character evaluation matrix. Therefore, the above system can be written compactly as \[ \left( \log_{\xi} \left( \mathcal{G}_1\odot_M\mathcal{G}_2^\top \right) \right)^\top \mathbf{n} = \mathbf{0}, \] where \[ \mathbf{n} = (n_1,n_2,\ldots,n_k)^\top. \] Suppose that \[ \mathrm{rk}_q\!\left( \log_{\xi} \left( \mathcal{G}_1\odot_M\mathcal{G}_2^\top \right) \right) = r. \] Then the solution space of the above system has dimension $k-r$. Since each solution vector $\mathbf{n}$ determines a unique codeword in $\mathcal{C}_1\cap\mathcal{C}_2^M$, it follows that \[ \dim_{\mathbb{F}_q} \left( \mathcal{C}_1\cap\mathcal{C}_2^M \right) = k-r. \] Therefore, \[ \dim_{\mathbb{F}_q} \left( \mathcal{C}_1\cap\mathcal{C}_2^M \right) = \ell \quad\Longleftrightarrow\quad r=k-\ell, \] which establishes the first assertion. For the second statement, let \[ \mathbf{y}\in \mathcal{C}_1\cap\mathcal{C}_2^M. \] Since $\mathbf{y}\in\mathcal{C}_2^M$, it can be expressed as an $\mathbb{F}_q$-linear combination of the rows of $\mathcal{H}_2$, namely \[ \mathbf{y} = \sum_{i=1}^{nm-k} n_i\mathbf{H}_{2i}. \] Imposing the additional condition $\mathbf{y}\in\mathcal{C}_1$ leads to a homogeneous linear system whose coefficient matrix is \[ \log_{\xi} \left( \mathcal{H}_1\odot_M\mathcal{H}_2^\top \right). \] Repeating the same argument as above, we obtain \[ \mathrm{rk}_q \left( \log_{\xi} \left( \mathcal{H}_1\odot_M\mathcal{H}_2^\top \right) \right) = nm-k-\ell. \] Finally, the equalities involving \[ \log_{\xi} \left( \mathcal{G}_2\odot_M\mathcal{G}_1^\top \right) \] and \[ \log_{\xi} \left( \mathcal{H}_2\odot_M\mathcal{H}_1^\top \right) \] follow immediately by symmetry. This completes the proof. \end{proof}
	}
	
	\begin{cor}
		For $i=1,2$, let $\CC_i$ be additive $(n, q^k)$ codes over $\FF_{q^m}$. Then
		\[
		\mathrm{rk}_q(\CC_1\cap \CC_2^M)=\mathrm{rk}_q(\CC_2\cap \CC_1^M).
		\]
	\end{cor}
	
	\begin{proof}
		For  $i=1,2$, suppose that $\G_i$ is a generator matrix of $\CC_i$. By applying Theorem~\ref{th-3.1}, we get $\mathrm{rk}_q\left(\log_{\xi}\left(\G_1\odot_{M} \G_2^\top \right)\right)=k-\mathrm{rk}_q\left(\CC_1\cap\CC_2^M\right)$ and $\mathrm{rk}_q\left(\log_{\xi}\left(\G_2\odot_{M} \G_1^\top \right)\right)=k-\mathrm{rk}_q\left(\CC_2\cap\CC_1^M\right)$. Since $\mathrm{rk}_q\left(\log_{\xi}\left(\G_1\odot_{M} \G_2^\top \right)\right)=\mathrm{rk}_q\left(\log_{\xi}\left(\G_2\odot_{M} \G_1^\top \right)\right)$, hence the desired result follows immediately.
	\end{proof}
	
	Consequently, the following statements provide precise characterizations of ACP and ACD for additive codes over finite fields in terms of the invertibility of certain associated matrices.
	
	\begin{cor}
		For $i=1,2$, let $\CC_i$ be additive $(n, q^k)$ codes over $\FF_{q^m}$. Then the pair $\{\CC_1, \CC_2^M\}$ is an ACP if and only if
		\[
		\log_{\xi}\left(\G_1\odot_{M} \G_2^\top \right)
		\]
		is invertible over $\FF_q$.
	\end{cor}
	
	\begin{cor}
		Let $\CC$ be an additive $(n, q^k)$ code over $\FF_{q^m}$. Then $\CC$ is an ACD if and only if
		\[
		\log_{\xi}\left(\G\odot_{M} \G^\top \right)
		\]
		is invertible over $\FF_q$.
	\end{cor}
	{
		\begin{thm}\label{th-3.2} For $i=1,2$ with $i\neq j$, let $\CC_i$ be an additive $(n, q^k)$ code over $\mathbb{F}_{q^m}$ with a generator matrix $\mathcal{G}_i$ and a parity-check matrix $\mathcal{H}_i$. If \[ \mathrm{rk}_q \begin{pmatrix} \mathcal{G}_i\\ \mathcal{H}_j \end{pmatrix} =nm-\ell, \] then \[ \mathrm{rk}_q\left(\CC_1\cap\CC_2^M\right)\ge \ell. \] \end{thm} \begin{proof} We prove the statement for $i=1$ and $j=2$, since the case $i=2$ and $j=1$ follows analogously. Let \[ R_{\G_2}=\operatorname{row}_{\mathbb F_q}(\G_2) \qquad\text{and}\qquad R_{\HH_1}=\operatorname{row}_{\mathbb F_q}(\HH_1) \] denote the $\mathbb F_q$-row spaces of $\G_2$ and $\HH_1$, respectively. Since $\CC_2$ has cardinality $q^k$ and $\CC_1^M$ has cardinality $q^{mn-k}$, we have \[ \dim(R_{\G_2})=k \qquad\text{and}\qquad \dim(R_{\HH_1})=mn-k. \] The assumption \[ \mathrm{rk}_q \begin{pmatrix} \mathcal{G}_1\\ \mathcal{H}_2 \end{pmatrix} =nm-\ell \] implies that \[ \dim(R_{\G_2}+R_{\HH_1})=mn-\ell. \] Therefore, \[ \begin{aligned} \dim(R_{\G_2}\cap R_{\HH_1}) &=\dim(R_{\G_2})+\dim(R_{\HH_1}) -\dim(R_{\G_2}+R_{\HH_1})\\ &=k+(mn-k)-(mn-\ell)\\ &=\ell. \end{aligned} \] Since $\dim(R_{\G_2}\cap R_{\HH_1})=\ell$, by performing invertible row operations on $\G_2$ (which leave the code $\CC_2$ unchanged), we may assume that \[ \widetilde{\G}_2= \begin{pmatrix} \G_{2,1}\\ \G_{2,2} \end{pmatrix}, \] where $\G_{2,1}$ is an $\ell\times n$ matrix whose rows form a basis of $R_{\G_2}\cap R_{\HH_1}$, and $\G_{2,2}$ is a $(k-\ell)\times n$ matrix. Since every row of $\G_{2,1}$ belongs to $R_{\HH_1}$, each of them can be expressed as an $\mathbb F_q$-linear combination of the rows of $\HH_1$. Consequently, for every $\bx\in\CC_1$, we have \[ \HH_1\odot_M\bx^\top=\mathbf{1}, \] which yields \[ \G_{2,1}\odot_M\bx^\top = \mathbf{1}_{\ell\times1}. \] Taking $\bx$ successively as each row of the generator matrix $\G_1$, whose rows form an $\mathbb F_q$-basis of $\CC_1$, we obtain \[ \G_{2,1}\odot_M\G_1^\top = \mathbf{1}_{\ell\times k}. \] Applying the entrywise logarithm gives \[ \log_{\xi} \left( \G_{2,1}\odot_M\G_1^\top \right) = \mathbf{0}_{\ell\times k}. \] Hence, \[ \log_{\xi} \left( \widetilde{\G}_2\odot_M\G_1^\top \right) = \begin{pmatrix} \log_{\xi}(\G_{2,1}\odot_M\G_1^\top)\\ \log_{\xi}(\G_{2,2}\odot_M\G_1^\top) \end{pmatrix} = \begin{pmatrix} \mathbf{0}_{\ell\times k}\\ \log_{\xi}(\G_{2,2}\odot_M\G_1^\top) \end{pmatrix}. \] Therefore, the first $\ell$ rows of \[ \log_{\xi} \left( \widetilde{\G}_2\odot_M\G_1^\top \right) \] are identically zero, and thus \[ \mathrm{rk}_q \left( \log_{\xi} \left( \widetilde{\G}_2\odot_M\G_1^\top \right) \right) \le k-\ell. \] Since $\widetilde{\G}_2$ is obtained from $\G_2$ by invertible row operations, the above rank coincides with \[ \mathrm{rk}_q \left( \log_{\xi} \left( \G_2\odot_M\G_1^\top \right) \right), \] and hence \[ \mathrm{rk}_q \left( \log_{\xi} \left( \G_2\odot_M\G_1^\top \right) \right) \le k-\ell. \] Finally, Theorem~\ref{th-3.1} implies that \[ \mathrm{rk}_q(\CC_1\cap\CC_2^M)\ge \ell. \] This completes the proof. \end{proof}}
	
	{
		
		\begin{thm}\label{th-3.3} For $i, j \in \{1, 2\}$ with $i \neq j$, let $\mathcal{C}_i$ be an additive $(n, q^k)$ code over $\mathbb{F}_{q^m}$. If \[ \mathrm{rk}_q\left(\mathcal{C}_1 \cap \mathcal{C}_2^M\right)=\ell, \] then there exist a generator matrix $\mathcal{G}_i$ and a parity-check matrix $\mathcal{H}_j$ of $\mathcal{C}_i$ and $\mathcal{C}_j$, respectively, such that \[ \mathrm{rk}_q \begin{pmatrix} \mathcal{G}_i\\ \mathcal{H}_j \end{pmatrix} = nm-\ell. \] \end{thm} \begin{proof} Assume that \[ \mathrm{rk}_q\left(\mathcal{C}_1\cap\mathcal{C}_2^M\right)=\ell. \] Let \[ \mathcal{B}_0=\{\mathbf{r}_1,\mathbf{r}_2,\ldots,\mathbf{r}_\ell\} \] be an $\mathbb{F}_q$-basis of the intersection subspace \[ \mathcal{C}_1\cap\mathcal{C}_2^M. \] Since $\mathcal{C}_1\cap\mathcal{C}_2^M$ is an $\mathbb{F}_q$-subspace of both $\mathcal{C}_1$ and $\mathcal{C}_2^M$, the basis extension theorem allows us to extend $\mathcal{B}_0$ to bases of these two spaces. More precisely, we may choose \[ \mathcal{B}_1= \{\mathbf{r}_1,\mathbf{r}_2,\ldots,\mathbf{r}_\ell, \mathbf{r}_{\ell+1},\ldots,\mathbf{r}_k\} \] to be an $\mathbb{F}_q$-basis of $\mathcal{C}_1$, and \[ \mathcal{B}_2= \{\mathbf{r}_1,\mathbf{r}_2,\ldots,\mathbf{r}_\ell, \mathbf{s}_{\ell+1},\ldots,\mathbf{s}_{mn-k}\} \] to be an $\mathbb{F}_q$-basis of $\mathcal{C}_2^M$. Let $\mathcal{G}_1$ be a generator matrix of $\mathcal{C}_1$ whose rows are precisely the vectors of $\mathcal{B}_1$. Likewise, let $\mathcal{H}_2$ be a parity-check matrix of $\mathcal{C}_2$. Since $\mathcal{H}_2$ is a generator matrix of $\mathcal{C}_2^M$, we may choose its rows to be exactly the vectors of $\mathcal{B}_2$. The row spaces of $\mathcal{G}_1$ and $\mathcal{H}_2$ therefore intersect in a subspace of dimension $\ell$, namely \[ \operatorname{span}_{\mathbb F_q}(\mathcal{B}_0). \] Consequently, \[ \begin{aligned} \mathrm{rk}_q \begin{pmatrix} \mathcal{G}_1\\ \mathcal{H}_2 \end{pmatrix} &= \rk_q(\mathcal{C}_1+\mathcal{C}_2^M)\\ &= \rk_q(\mathcal{C}_1)+\rk_q(\mathcal{C}_2^M) -\rk_q(\mathcal{C}_1\cap\mathcal{C}_2^M)\text{ by Remark~\ref{rk-1} }\\ &= k+(nm-k)-\ell\\ &= nm-\ell. \end{aligned} \] This proves the desired equality. The argument for the pair $(\mathcal{G}_2,\mathcal{H}_1)$ is identical. Hence, \[ \mathrm{rk}_q \begin{pmatrix} \mathcal{G}_i\\ \mathcal{H}_j \end{pmatrix} = nm-\ell, \] which completes the proof. \end{proof}}
	
	\begin{remark}
		We emphasize that if \(\ell(=k)=\min\{k,nm-k\}\), by Theorems~\ref{th-3.2}, \ref{th-3.3} and Lemma~\ref{lm-3.1}, we conclude that \(\mathrm{rk}_q\left(\CC_1 \cap \CC_2^M\right) = \ell\) if and only if there exist a generator matrix \(\mathcal{G}_i\) and a parity-check matrix \(\HH_j\) of \(\CC_i\) such that \(\mathrm{rk}_q\left(\begin{array}{c} \G_i \\ \HH_j \end{array}\right) = nm - \ell\) with \(i \neq j\). 
	\end{remark}
	
	{ A choice of duality matrix $M$ is defined as \textit{skew-symmetric} if its character evaluations satisfy
		$$\chi_u(v) = \chi_v(u)^{-1} \quad \text{for all } u, v \in G.$$
		When we apply the entry-wise discrete logarithm mapping base $\xi$ (where $\xi$ is a primitive $p$-th root of unity), then we get
		$$\log_{\xi}(\chi_u(v)) = -\log_{\xi}(\chi_v(u)).$$
	}
	
	The following theorem demonstrates the existence of non-zero \(\ell\)-RAIP of codes.
	\begin{thm}\label{th-3.8}
		Let $\mathcal{C}$ be an additive $(n, q^k)$ code over $\mathbb{F}_{q^m}$ with with an odd dimension $k$. If $q$ is odd and $M$ is skew-symmetric, then the pair $\{\mathcal{C}, \mathcal{C}^M\}$ is a non-zero $\ell$-RAIP of codes.
	\end{thm}
	
	{\begin{proof}
			Let $\mathcal{G}$ be a generator matrix of $\mathcal{C}$ and set \[ S=\log_{\xi}\!\left(\mathcal{G}\odot_M\mathcal{G}^\top\right) \in \mathbb{M}_{k\times k}(\mathbb{F}_q). \] Since $M$ is skew-symmetric, we have \[ \chi_u(v)=\chi_v(u)^{-1} \qquad\text{for all }u,v\in\mathbb{F}_{q^m}. \] Consequently, \[ S_{i,j} = \log_{\xi}\!\bigl(\chi_{\mathbf{g}_i}(\mathbf{g}_j)\bigr) = -\log_{\xi}\!\bigl(\chi_{\mathbf{g}_j}(\mathbf{g}_i)\bigr) = -S_{j,i}. \] Hence, \[ S^\top=-S, \] showing that $S$ is a skew-symmetric matrix over $\mathbb{F}_q$. Since $k$ is odd, every skew-symmetric matrix of order $k$ over $\mathbb{F}_q$ has zero determinant. Therefore, \[ \det(S)=0, \] and thus $S$ is singular. It follows that \[ \mathrm{rk}_q(S)\le k-1. \] Applying Theorem~\ref{th-3.1}, we obtain \[ \mathrm{rk}_q(\mathcal{C}\cap\mathcal{C}^M) = k-\mathrm{rk}_q(S) \ge 1. \] This completes the proof. 
	\end{proof}}
	

We aim to construct an Additive Complementary Pair (ACP) of codes equivalent to the original codes. To achieve this, we formalize the coordination mapping. Let $\mathcal{B} = \{\alpha_1, \alpha_2, \ldots, \alpha_m\}$ be a fixed $\mathbb{F}_q$-basis for the extension field $\mathbb{F}_{q^m}$. We define the mapping $\Phi : \mathbb{F}_q^{mn} \rightarrow \mathbb{F}_{q^m}^n$ by
\begin{equation}\label{eq-mapping}
	\Phi(\mathbf{a}) = \left( \sum_{j=1}^m a_{1,j}\alpha_j, \sum_{j=1}^m a_{2,j}\alpha_j, \ldots, \sum_{j=1}^m a_{n,j}\alpha_j \right),
\end{equation}
where $\mathbf{n} = (a_{1,1}, \dots, a_{1,m}, \dots, a_{n,1}, \dots, a_{n,m}) \in \mathbb{F}_q^{mn}$. 

By construction, $\Phi$ is an $\mathbb{F}_q$-linear isomorphism, which implies the existence of inverse mapping $\Phi^{-1} : \mathbb{F}_{q^m}^n \rightarrow \mathbb{F}_q^{mn}$. Consequently, a subset $C \subseteq \mathbb{F}_q^{mn}$ is a linear $[mn, k]$ code over $\mathbb{F}_q$ if and only if its image $\Phi(C) \subseteq \mathbb{F}_{q^m}^n$ is an $\mathbb{F}_q$-linear additive $(n, q^k)$ code over $\mathbb{F}_{q^m}$. Conversely, for any additive $(n, q^k)$ code $\mathcal{C} \subseteq \mathbb{F}_{q^m}^n$, its preimage satisfies $\Phi^{-1}(\mathcal{C}) \in \mathbb{M}_{1 \times mn}(\mathbb{F}_q)$, forms a linear $[mn, k]$ code over $\mathbb{F}_q$.
We will introduce a helpful lemma for establishing results in the subsequent sequels.

\begin{lemma}\label{lm-3.1}
	For $i=1,2$, let $\CC_i$ be an additive $(n, q^k)$ code over $\FF_{q^m}$. Then $\Phi^{-1}\left(\CC_1\cap \CC_2^M\right)=\Phi^{-1}\left(\CC_1\right)\cap \Phi^{-1}\left(\CC_2^M\right)$.   
\end{lemma}
\begin{proof}
	Let $\bx\in \Phi^{-1}\left(\CC_1\cap \CC_2^M\right)$, then there is $\cc\in\CC_1\cap\CC_2^M$ such that $\bx=\Phi^{-1}(\cc)$. Since, $\cc\in\CC_1\cap\CC_2^M$, we have $\cc\in\CC_1$ and $\cc\in\CC_2^M$, which implies $\bx=\Phi^{-1}(\cc)\in\Phi^{-1}\left(\CC_1\right)\cap \Phi^{-1}\left(\CC_2^M\right)$. Thus, we conclude that $\Phi^{-1}\left(\CC_1\cap \CC_2^M\right)\subseteq\Phi^{-1}\left(\CC_1\right)\cap \Phi^{-1}\left(\CC_2^M\right)$.
	
	Conversely, let $\by \in\Phi^{-1}\left(\CC_1\right)\cap \Phi^{-1}\left(\CC_2^M\right)$, then there are $\cc \in\CC_1$ and  $\cc'\in\CC_2^M$ such that $\Phi^{-1}(\cc)=\Phi^{-1}(\cc')$. Since, $\Phi^{-1}$ is injective, we conclude that $\cc=\cc'$, which implies $\cc\in\CC_1\cap\CC_2^M$. Thus, $\by=\Phi^{-1}(\cc)\in\Phi^{-1}\left(\CC_1\cap \CC_2^M\right)$. Hence, $\Phi^{-1}\left(\CC_1\right)\cap \Phi^{-1}\left(\CC_2^M\right)\subseteq\Phi^{-1}\left(\CC_1\cap \CC_2^M\right)$, which completes the proof.    
\end{proof}

\begin{remark}\label{rk-1}
	Since the mapping $\Phi: \mathbb{F}_q^{mn} \rightarrow \mathbb{F}_{q^m}^n$ defined in \eqref{eq-mapping} is an $\mathbb{F}_q$-linear vector space isomorphism, it preserves the underlying subspace dimension. Consequently, by Lemma \ref{lm-3.1}, we obtain
	$$ \mathrm{rk}_q\left(\mathcal{C}_1 \cap \mathcal{C}_2^M\right) = \mathrm{dim}_q\left(\Phi^{-1}\left(\mathcal{C}_1\right) \cap \Phi^{-1}\left(\mathcal{C}_2^M\right)\right). $$
\end{remark}

\begin{thm}\label{th-3.32}
	For $i, j \in \{1, 2\}$ with $i \neq j$, let $\mathcal{C}_i$ be an additive $(n, q^k)$ code over $\mathbb{F}_{q^m}$ with a generator matrix $\mathcal{G}_i$ and a parity-check matrix $\mathcal{H}_i$. If 
	$ \mathrm{rk}_q \begin{pmatrix} \mathcal{G}_i \\ \mathcal{H}_j \end{pmatrix} = nm - \ell, $
	then $\mathrm{rk}_q\left(\mathcal{C}_1 \cap \mathcal{C}_2^M\right) = \ell$.   
\end{thm}

\begin{proof}
	Since \(\mathrm{rk}_{q}\left(\begin{array}{c}
		\G_1 \\
		\HH_2
	\end{array}\right) = nm - \ell\), it follows that \(\mathrm{rk}_q\left(\CC_1 + \CC_2^M\right) = nm - \ell\). Note that  \(\mathrm{rk}_q\left(\CC_1 + \CC_2^M\right) = \mathrm{rk}_q(\CC_1) + \mathrm{rk}_q(\CC_2^M) - \mathrm{rk}_q\left(\CC_1 \cap \CC_2^M\right)\) (by Remark~\ref{rk-1}). Since $\mathrm{rk}_q(\CC)=k$ and $\mathrm{rk}_q(\CC^M)=nm-k$, then \(\mathrm{ rk}_q\left(\CC_1 \cap \CC_2^M\right) = \ell\), which completes the proof.
\end{proof}

\begin{thm}
	For $i=1,2$, let $\CC_i$ be an additive $(n, q^k)$ code over $\FF_{q^m}$ with a generator matrix $\mathcal{G}_i$ and a parity check matrix $\HH_i$. Then $\mathrm{rk}_q\left(\CC_1\cap\CC_2^M\right)= \ell $ if and only if for $i\neq j$, $\mathrm{rk}_q\left(\begin{array}{c}
		\G_i \\
		\HH_j
	\end{array}\right)=nm-\ell$.    
\end{thm}
\begin{proof}
	The result follows immediately by applying Theorems \ref{th-3.3} and \ref{th-3.32}.   
\end{proof}
\begin{cor}
	For $i=1,2$, let $\CC_i$ be an additive $(n, q^k)$ code over $\FF_{q^m}$ with a generator matrix $\mathcal{G}_i$ and a parity check matrix $\HH_i$. Then the pair $\{\CC_1,\CC_2^M\}$ is ACP if and only if for $i\neq j$, the matrix $\left(\begin{array}{c}
		\G_i \\
		\HH_j
	\end{array}\right)$ is invertible.    
\end{cor}
\begin{cor}
	Let $\CC$ be an additive $(n, q^k)$ code over $\FF_{q^m}$ with a generator matrix $\mathcal{G}$ and parity check matrix $\HH$. Then the code $\CC$ is ACD if and only if the matrix $\left(\begin{array}{c}
		\G \\
		\HH
	\end{array}\right)$ is invertible.    
\end{cor}
Any \( n \times n \) monomial matrix can be expressed as \( \mathcal{M} = DP \), where \( D = \textit{diag}(d_1, d_2, \ldots, d_n) \) is a diagonal matrix with \( d_i \in \mathbb{F}_{q^m}^* \) for \( i = 1, 2, \ldots, n \), and \( P \) is an \( n \times n \) permutation matrix. A code \( \mathcal{C} \) is said to be monomially equivalent to another code \( \mathcal{C}' \) if there exists a monomial matrix \( \mathcal{M} \) such that \( \mathcal{D} = \mathcal{C} \mathcal{M} \).

\begin{thm}\label{th-3.4}
	For $i=1,2$, let $C_i$ be linear $[n,k_i]$ codes over $\mathbb{F}_q$ with $q \geq 3$ and $k_1 + k_2 = n$. If $\dim(C_1 \cap C_2) = l$, then the pair $\{C_1, C_2\}$ is monomially equivalent to an $\ell$-DLIP of codes, where $0 \leq \ell \leq l$.
\end{thm}
\begin{proof}
	Working in a similar manner as in Theorem 3.3 of Anderson et al. \cite{ACLMRS24}, the desired result follows.    
\end{proof}

The following theorem will characterize the ACP of codes.

\begin{thm}
	For $i=1,2$, let $\CC_i$ be an additive $(n, q^k)$ code over $\FF_{q^m}$, where $q \geq 3$. If $\mathrm{rk}_q\left(\CC_1\cap\CC_2^M\right) = l$, then the pair $\{\CC_1,\CC_2^M\}$ is monomially equivalent to an $\ell$-RAIP of codes, where $0 \leq \ell \leq l$.
	
\end{thm}

\begin{proof}
	To demonstrate the result, we first define \( C_1 = \Phi^{-1}(\mathcal{C}_1) \) and \( C_2 = \Phi^{-1}(\mathcal{C}_2^M) \). It is essential to note that \( C_1 \) is a linear \([nm, k]\) code, whereas \( C_2 \) is a linear \([mn, mn-k]\) code over \(\mathbb{F}_q\). Based on Remark~\ref{rk-1}, we can conclude that \( \dim(C_1 \cap C_2) = l \). Then, by applying Theorem~\ref{th-3.4}, we arrive at our desired result.
\end{proof}

\begin{remark}
	As a result, we conclude that for \( i = 1, 2 \), let \( \mathcal{C}_i \) be an additive \( (n, q^k) \) code over \( \mathbb{F}_{q^m} \) with \( q \geq 2 \). The pair \( \{\mathcal{C}_1, \mathcal{C}_2^M\} \) is monomially equivalent to an ACP of codes.
	
\end{remark}

In the upcoming sequel, we address the equivalence problems for ACD codes. Specifically, we present an efficient characterization of a given code to determine if it is monomially equivalent to an ACD code.

Note that for any nonzero scalar $b \in \mathbb F_q^*$,
\[
\chi_{{b}^{-1}{u}}({b}{v})
= \chi_{{u}}({v}).
\]

\begin{thm}\label{th-3.5}
	Let $\mathcal{C}$ be an additive $(n, q^k)$ code over $\mathbb{F}_{q^m}$ with $q \geq 4$. Then $\mathcal{C}$ is monomially equivalent to an additive complementary dual (ACD) code.
\end{thm}

\begin{proof}
	If $\mathrm{rk}_q\left(\mathcal{C} \cap \mathcal{C}^M\right) = 0$, the code is an ACD code by definition. Assume $\mathrm{rk}_q\left(\mathcal{C} \cap \mathcal{C}^M\right) = l > 0$. 
	
	Let $C = \Phi^{-1}(\mathcal{C})$ and $C^\perp = \Phi^{-1}(\mathcal{C}^M)$ be the corresponding preimages under the $\mathbb{F}_q$-linear vector space isomorphism $\Phi$. By the property of the mapping, $C$ is a linear $[mn, k]$ code and $C^\perp$ is a linear $[mn, mn-k]$ code over $\mathbb{F}_q$. From Remark \ref{rk-1}, the dimension of their intersection satisfies
	$$ \dim_q(C \cap C^\perp) = \mathrm{rk}_q\left(\mathcal{C} \cap \mathcal{C}^M\right) = l. $$
	
	Choose a basis of $C\cap C^\perp$ and extend to bases of $C$ and $C^\perp$. Consider generator matrices of $C$ and $C^\perp$ are
	\[
	G_1=\begin{pmatrix}
		I_l & 0 & P_1\\
		0 & I_{k-l} & P_2
	\end{pmatrix},\qquad
	G_2=\begin{pmatrix}
		I_l & 0 & P_1\\
		B_1 & B_2 & B_3
	\end{pmatrix},
	\]
	where $P_1\in\mathbb{F}_q^{l\times(mn-k)}$, $P_2\in\mathbb{F}_q^{(k-l)\times(mn-k)}$, $B_1\in\mathbb{F}_q^{(mn-k-l)\times l}$, $B_2\in\mathbb{F}_q^{(mn-k-l)\times(k-l)}$, $B_3\in\mathbb{F}_q^{(mn-k-l)\times(mn-k)}$. The first $l$ rows of both matrices form a basis of $C\cap C^\perp$.
	
	Let $\lambda_1,\dots,\lambda_l\in\mathbb{F}_q^*$ satisfy $\lambda_i^2\neq1$. Define the diagonal matrix
	\[
	D=\operatorname{diag}(\lambda_1,\dots,\lambda_l,1,\dots,1)\in\mathbb{F}_q^{mn\times mn}.
	\]
	Then $DC$ and $D^{-1}C^\perp$ have generator matrices
	\[
	G_1'=\begin{pmatrix}
		\Lambda & 0 & P_1\\
		0 & I_{k-l} & P_2
	\end{pmatrix},\qquad
	G_2'=\begin{pmatrix}
		\Lambda^{-1} & 0 & P_1\\
		\Lambda^{-1}B_1 & B_2 & B_3
	\end{pmatrix},
	\]
	where $\Lambda=\operatorname{diag}(\lambda_1,\dots,\lambda_l)$ and $\Lambda^{-1}=\operatorname{diag}(\lambda_1^{-1},\dots,\lambda_l^{-1})$.
	
	Consider the $mn\times mn$ matrix
	\[
	\mathcal{M}=\begin{pmatrix}G_1'\\ G_2'\end{pmatrix}.
	\]
	Hence
	\[
	\operatorname{rk}(\mathcal{M})=\operatorname{rk}\begin{pmatrix}
		\Lambda-\Lambda^{-1} & 0 & 0\\
		0 & I_{k-l} & P_2\\
		\Lambda^{-1}B_1 & B_2 & B_3
	\end{pmatrix}.
	\]
	
	{
		The determinant of $\mathcal{M}$ is nonzero if and only if
		$\lambda_i-\lambda_i^{-1}\neq 0,$ for all $ 1\le i\le l,$
		or equivalently,
		$\lambda_i^2\neq 1,
		\text{ for all } 1\le i\le l.$
		Over a finite field $\mathbb{F}_q$, the equation
		$x^2-1=0$
		has at most two solutions, namely $x=1$ and $x=-1$. Therefore, each $\lambda_i$ must be chosen from $\mathbb{F}_q^\ast\setminus\{1,-1\}$. Since $|\mathbb{F}_q^\ast|=q-1$, the existence of such choices requires
		$q-1>2,$
		which is equivalent to
		$q>3.$}
	
	As \( q \geq 4 \), we choose \( \lambda_i \in \mathbb{F}_q^* \) such that \( \lambda_i^2 \neq 1 \) for all \( 1 \leq i \leq l \) such that
	\[
	\mathrm{rk}_q\left(\begin{array}{c} 
		\mathcal{G}_{\mathbf{a}}\\
		\mathcal{G}_{\mathbf{a}^{-1}}
	\end{array}\right) = mn.
	\]
	
	Since \( \Phi \) is an isomorphism, there exists \( \mathbf{b} \in (\mathbb{F}_{q^m}^*)^n \) such that \( \Phi(\mathbf{a}) = \mathbf{b} \). Therefore, we have 
	
	\[
	\mathrm{rk}_q\left(\mathbf{b}\mathcal{C} \cap \mathbf{b}^{-1} \mathcal{C}^M\right) = 0.
	\]
	
	Furthermore, \( \mathbf{b}^{-1}\mathcal{C}^M \) is the dual code of \( \mathcal{C} \) with respect to \( M \) because 
	
	\[
	\chi_{\mathbf{b}^{-1}\mathbf{u}}(\mathbf{b}\mathbf{v})=\prod_{i=1}^n\chi_{{b}_i^{-1}{u}_i}({b}_i{v}_i) = \chi_{\mathbf{u}}(\mathbf{v}) = 1,
	\]
	for all \(\mathbf{u} \in \mathcal{C}^M\) and \(\mathbf{v} \in \mathcal{C}\). Hence, the result follows immediately.
\end{proof}

\begin{remark}
	Suppose $\CC$ is an additive code with its dual $\CC^{M}$ under $M$. Then, $\Phi(\CC^{M})$ is the dual code of $\Phi(\CC)$ under some suitable $\Phi$. For example, let $\FF_{2^2}=\{0, 1, \omega, \omega^2\}$. Let $\CC=\langle(0, 1) \rangle=\{(0, 0), (0, 1) \}$ and consider the duality $M$ as
	$$\begin{array}{|c|c|c|c|c|}
		\hline
		M & 0 & 1 & \omega & \omega^2 \\
		\hline
		0 & 1 & 1 & 1 & 1  \\
		1 & 1 & 1 & -1 & -1  \\
		\omega & 1 & -1 & -1 & 1 \\
		\omega^2 & 1 & -1 & 1 & -1 \\
		\hline
	\end{array}.$$ 
	For $\Phi_1(1)=(1,0)$, $\Phi_1(\omega)=(0,1)$, we have
	\[
	\Phi_1(\mathcal C) = \operatorname{span}_{\mathbb F_2}\{(0,0,1,0)\},
	\]
	so
	\[
	(\Phi_1(\mathcal C))^\perp = \{(x_1,x_2,x_3,x_4)\in\mathbb F_2^4 \mid x_3=0\},
	\]
	whereas
	\[
	\Phi_1(\mathcal C^M) = \{(x_1,x_2,x_3,0)\in\mathbb F_2^4 \mid x_i\in\mathbb F_2\}.
	\]
	Hence $\Phi_1(\mathcal C^M) \neq (\Phi_1(\mathcal C))^\perp$.
	
	For $\Phi_2(1)=(1,1)$, $\Phi_2(\omega)=(0,1)$, one verifies that $\operatorname{Tr}(uv) = \langle \Phi_2(u), \Phi_2(v) \rangle$ for all $u,v$. Then
	\[
	\Phi_2(\mathcal C) = \operatorname{span}_{\mathbb F_2}\{(0,0,1,1)\},
	\]
	and
	\[
	(\Phi_2(\mathcal C))^\perp = \{(x_1,x_2,x_3,x_3)\in\mathbb F_2^4 \mid x_i\in\mathbb F_2\}.
	\]
	Moreover,
	\[
	\Phi_2(\mathcal C^M) = \{(x_1,x_2,x_3,x_3)\in\mathbb F_2^4 \mid x_i\in\mathbb F_2\},
	\]
	so $\Phi_2(\mathcal C^M) = (\Phi_2(\mathcal C))^\perp$.
	
	Thus, the proof of Theorem~\ref{th-3.5} requires fixing a compatible $\Phi$. Furthermore, Theorem 3.16 can be derived from \cite[Corollary~14]{CMT18}.
\end{remark}

\section{Construction of $\ell$-RAIP Codes}\label{sec:4}

The main goal of this section is to construct $\ell$-RAIP codes for given additive codes $\CC_i$, where $i=1,2$.

\begin{thm}
	Let $\CC_1=\langle \textbf{u}_1, \textbf{u}_2, \ldots, \textbf{u}_r\rangle$ and $\CC_1=\langle \textbf{v}_1, \textbf{v}_2, \ldots, \textbf{v}_r\rangle$ be two additive codes of length $n \equiv 0 \pmod q$ over $\FF_{q^{e}}=\{\langle\beta_1, \beta_2, \ldots, \beta_{e}\rangle\}$, where $\textbf{u}_i=(\beta_i, \beta_i, \ldots, \beta_i)$, for $1\leq i \leq r$ and $\textbf{v}_{i-r}=(\alpha_i, \alpha_i, \ldots, \alpha_i)$, for $r+1\leq i \leq 2r\leq e$. Then the pair $\{\CC_1,\CC_2^M\}$ has $r$-RAIP of codes.  
\end{thm}

\begin{proof}
	
	To prove the result, suppose that \(\bx \in \CC_1 \cap \CC^M_2\). This implies that \(\bx = \sum_{i=1}^{r} n_i \textbf{u}_i\), which belongs to both \(\CC_1\) and \(\CC_2^M\). Therefore, we have \(\chi_{(\sum_{i=1}^{r} n_i \textbf{u}_i)}(\textbf{v}_j) = 1\) for all \(j = 1, 2, \ldots, k\).
	
	Consequently, it follows that 
	
	\[
	\prod_{i=1}^{r} \left(\chi_{\textbf{u}_i}(\textbf{v}_j)\right)^{n_i} = \prod_{i=1}^{r} \left(\chi_{\beta_i}(\beta_j)\right)^{n_i} = 1.
	\]
	
	Now, assume that \(\chi_{\beta_i}(\beta_j) = \xi^{s_{ij}}\) for \(1 \leq i, j \leq r\), where \(s_{ij} \in \FF_q\) and \(\xi\) is a primitive \(p\)-th root of unity. This leads to the expression \(\xi^{n\sum_{i=1}^{r} s_{ij} n_i} = 1\),
	which implies that 
	
	\(n\sum_{i=1}^{r} s_{ij} n_i = 0 \mod q\).
	
	Thus, we arrive at a system of linear equations as follows
	
	$$n\left( {\begin{array}{ccccc}
			s_{11} & s_{21}&\cdots & s_{k1} \\
			s_{12} & s_{22}&\cdots & s_{k2} \\
			\vdots & \vdots&\ddots & \vdots \\
			s_{1k} & s_{2k}&\cdots & s_{kk}
	\end{array} } \right) \left( {\begin{array}{c}
			n_1  \\
			n_2  \\
			\vdots  \\
			n_k 
	\end{array} } \right)=\left( {\begin{array}{c}
			0 \\
			0  \\
			\vdots  \\
			0 
	\end{array} } \right).$$ Consequently, we get
	$$n\left(log_{\xi}\left(\G_1\odot_{M} \G_2^\top \right)\right)^\top\left( {\begin{array}{c}
			n_1  \\
			n_2  \\
			\vdots  \\
			n_k 
	\end{array} } \right)=\left( {\begin{array}{c}
			0 \\
			0  \\
			\vdots  \\
			0 
	\end{array} } \right),$$ where $\G_i$ is the generator matrix of $\CC_i$, for $i=1,2$. As, $n\equiv 0\pmod q$, so the rank of the matrix $n\left(log_{\xi}\left(\G_1\odot_{M} \G_2^\top \right)\right)$ is equal to zero. Thus, the results follow immediately from Theorem~\ref{th-3.1}.  
\end{proof}

\begin{remark}
	Note that, if $n\not\equiv 0\pmod q$ in the above theorem, whether RAIP of code $\{\CC_1, \CC_2^M\}$  depends on the matrix $n\left( {\begin{array}{ccccc}
			s_{11} & s_{21}&\cdots & s_{k1} \\
			s_{12} & s_{22}&\cdots & s_{k2} \\
			\vdots & \vdots&\ddots & \vdots \\
			s_{1k} & s_{2k}&\cdots & s_{kk}
	\end{array} } \right)$. However, in the above theorem if $\CC$ is an additive code of length $n\not\equiv 0\pmod q$ with rank one and $\chi_{1}(\beta)\neq 1$ for some $\beta\in \FF_q^{*}$, then the pair $\{\CC_1, \CC_2^M\}$ forms an ACP of codes. For example, let $\FF_{2^2}=\langle  1, \lambda\rangle$. Let $\CC=\langle(1, 1, 1, 1) \rangle$, $\CC_2=\langle \lambda, \lambda, \lambda, \lambda\rangle$ and consider the duality $M$ as $\chi_{1}(1)=1$, $\chi_{1}(\lambda)=\xi$, $\chi_{\lambda}(1)=\xi^2$, and $\chi_{\lambda}(\lambda)=1$. It is easy to see that $\{\CC_1,\CC_2^M\}$ is an ACP of codes. 
\end{remark}
\begin{thm}\label{th-4.2}
	Let \(\CC\) be a self-orthogonal additive \((n, q^k)\) code over \(\FF_q\) with a generator matrix \(\G\). Additionally, let \(\CC_1\) and \(\CC_2\) be additive \((n+k, q^k)\) codes with generator matrices \(\G_1 = [\mathrm{diag}(x_1, x_2, \ldots, x_k) | \G]\) and \(\G_2 = [\mathrm{diag}(y_1, y_2, \ldots, y_k) | \G]\), respectively. 
	If \(\chi_{x_i}(y_i) = 1\) for \(1 \leq i \leq \ell\) and \(\chi_{x_i}(y_i) \neq 0\) for \(\ell + 1 \leq i \leq k\), then the pair \(\{\CC_1, \CC_2\}\) is an \(\ell + 1\)-RAIP of codes.
\end{thm}
\begin{proof}
	To prove the theorem, we assume $\G_i$ be a row of $\G$, for $i=1,2,\ldots,k$.
	Now, the matrix $\G_1\odot\G_2^\top$ can be written as 
	$$\G_1\odot\G_2^\top=\left(\begin{array}{cccc}
		\chi_{x_1}(y_1)\chi_{\G_{1}}(\G_{1}) & \chi_{\G_{1}}(\G_{2}) & \cdots & \chi_{\G_{1}}(\G_{k}) \\
		\chi_{\G_{2}}(\G_{1}) & \chi_{x_2}(y_2)\chi_{\G_{2}}(\G_{2}) & \cdots & \chi_{\G_{2}}(\G_{k}) \\
		\vdots & \vdots & \ddots & \vdots \\
		\chi_{\G_{k}}(\G_{1}) & \chi_{\G_{k}}(\G_{2}) & \cdots & \chi_{x_k}(y_k)\chi_{\G_{k}}(\G_{k}) \\
	\end{array}\right).$$
	Since, $\CC$ is self-orthogonal, then $\chi_{\G_i}(\G_j)=1$, for $1\leq i,j \leq k$. By hypothesis, we assume $\chi_{x_i}(y_i)=\xi^{\beta_{ij}}$, for $\ell+1\leq i,j \leq k$ and $\beta_{ij}\in\FF_q^*$. Consequently, we have
	$$\textit{log}_{\xi}\left(\G_1\odot\G_2^\top\right)=\left(\begin{array}{ccccccc}
		0 & 0 & \cdots & 0 & 0 & \cdots & 0 \\
		0 & 0 & \cdots & 0 & 0 & \cdots & 0 \\
		\vdots & \vdots & \ddots & \vdots & \vdots & \ddots & \vdots  \\
		0 & 0 & \cdots & 0 & 0 & \cdots & 0 \\
		0 & 0 & \cdots & 0 & \alpha_{\ell+1,\ell+1} & \cdots & 0 \\
		\vdots & \vdots & \ddots & \vdots & \vdots & \ddots & \vdots  \\
		0 & 0 & \cdots & 0 & 0 & \cdots & \alpha_{k,k} \\
	\end{array}\right).$$ Hence, the result follows immediately from Theorem~\ref{th-3.1}.
\end{proof}

{\begin{example}\label{ex-raip-computation}
		Let $\mathbb{F}_{2^3}$ be a finite extension field with basis $\{1, \lambda, \lambda^2\}$ over $\mathbb{F}_2$. Consider the self-orthogonal linear code $\mathcal{C}$ of length $n = 4$ generated by
		\[
		\mathcal{G} = \begin{pmatrix}
			1 & 1 & 1 & 1 \\
			\lambda & 0 & \lambda & \lambda \\
			\lambda & \lambda & 0 & \lambda
		\end{pmatrix}.
		\]
		Take $x_1 = \lambda$, $x_2 = 1$, $x_3 = 1$, and $y_1 = y_2 = y_3 = \lambda$. Let $\mathcal{C}_1$ and $\mathcal{C}_2$ be two additive $(7, 2^3)$ codes generated by $\mathcal{G}_1$ and $\mathcal{G}_2$, respectivly,
		\[
		\mathcal{G}_1 = \begin{pmatrix}
			\lambda & 0 & 0 & 1 & 1 & 1 & 1 \\
			0 & 1 & 0 & \lambda & 0 & \lambda & \lambda \\
			0 & 0 & 1 & \lambda & \lambda & 0 & \lambda
		\end{pmatrix}, \quad 
		\mathcal{G}_2 = \begin{pmatrix}
			\lambda & 0 & 0 & 1 & 1 & 1 & 1 \\
			0 & \lambda & 0 & \lambda & 0 & \lambda & \lambda \\
			0 & 0 & \lambda & \lambda & \lambda & 0 & \lambda
		\end{pmatrix}.
		\]
		
		Let $\xi = -1$ be the primitive second root of unity over $\mathbb{F}_2$. The duality character relation $M: \mathbb{F}_{2^3} \times \mathbb{F}_{2^3} \to \mathbb{C}^*$ satisfies $\chi_{1}(1) = 1$, $\chi_{1}(\lambda) = \xi = -1$, $\chi_{\lambda}(1) = \xi^2 = 1$, and $\chi_{\lambda}(\lambda) = 1$. 
		
		Now, 
		$$ \mathcal{G}_1 \odot \mathcal{G}_2^\top = \begin{pmatrix}
			1 & \xi & \xi \\
			1 & \xi & 1 \\
			1 & 1 & \xi
		\end{pmatrix}. $$
		Thus,
		$$ \log_{\xi}\left(\mathcal{G}_1 \odot \mathcal{G}_2^\top\right) = \begin{pmatrix}
			0 & 1 & 1 \\
			0 & 1 & 0 \\
			0 & 0 & 1
		\end{pmatrix}. $$
		
		By Theorem \ref{th-3.1}, we get $\ell = k - \mathrm{rk} = 3 - 2 = 1$. Thus, the pair $\{\mathcal{C}_1, \mathcal{C}_2^M\}$ forms a $(1+1) = 2$-RAIP of codes, verifying the Theorem~\ref{th-4.2}.
\end{example}}

\begin{thm}\label{th-4.4}
	For $i=1,2$, let $\CC_i$ be an additive $(n, q^k)$ code over $\FF_{q^m}$ with a generator matrix $\mathcal{G}_i$. Also, let let $\CC_i'$ be an additive $(n, q^{k+1})$ code over $\FF_{q^m}$ with a generator matrix $\mathcal{G}_i'=\left(\begin{array}{cc}
		\ba_i  \\
		\G_i 
	\end{array}\right)$, where $\ba_1\in\CC_2^M$, $\ba_2\in\CC_1^M$, and $\chi_{\ba_1}(\ba_2)=1$. If the pair $\{\CC_1,\CC_2\}$ is $\ell$-RAIP of codes, then the pair $\{\CC_1',\CC_2'\}$ is $\ell+1$-RAIP of codes. 
\end{thm}
\begin{proof}
	For the proof of the theorem, we assume $\G_{li}$ be a row of $G_{l}$, for $l=1,2$.
	Now, the matrix $\G_1'\odot_{M}\G_2'^\top$ can be written as 
	$$\G_1'\odot_{M}\G_2'^\top=\left(\begin{array}{cccc}
		\chi_{\ba_1}(\ba_2) & \chi_{\ba_{1}}(\G_{21}) & \cdots & \chi_{\ba_{1}}(\G_{2k}) \\
		\chi_{\G_{12}}(\ba_{2}) &  &  &  \\
		\vdots &  & \G_1'\odot_{M}\G_2'^\top  &  \\
		\chi_{\G_{1k}}(\ba_{2}) &  &  &  \\
	\end{array}\right).$$  
	Consequently, we have 
	$$\G_1'\odot_{M}\G_2'^\top=\left(\begin{array}{cccc}
		1 & 1 & \cdots & 1 \\
		1 &  &  &  \\
		\vdots &  & \G_1'\odot_{M}\G_2'^\top  &  \\
		1 &  &  &  \\
	\end{array}\right).$$
	Thus, the result follows from Theorem~\ref{th-3.1}.      
\end{proof}
{\begin{example}\label{ex-raip-extension}
		Let $\mathbb{F}_{2^3}$ be a finite extension field with basis $\{1, \lambda, \lambda^2\}$ over $\mathbb{F}_2$. Consider two additive codes $\mathcal{C}_1$ and $\mathcal{C}_2$ of length $n = 3$ and dimension $k = 1$ over $\mathbb{F}_{2^3}$ generated by 
		\[
		\mathcal{G}_1 = \begin{pmatrix} 1 & 1 & 1 \end{pmatrix}, \quad \mathcal{G}_2 = \begin{pmatrix} \lambda & \lambda & \lambda \end{pmatrix}.
		\]
		Let $\xi = -1$ be the primitive second root of unity over $\mathbb{F}_2$. We define the character duality relation $M: \mathbb{F}_{2^3} \times \mathbb{F}_{2^3} \to \mathbb{C}^*$ by setting $\chi_{1}(1) = \xi = -1$, $\chi_{1}(\lambda) = 1$, $\chi_{\lambda}(1) = 1$, and $\chi_{\lambda}(\lambda) = \xi^2 = 1$.
		
		We first compute 
		$$ \mathcal{G}_1 \odot_M \mathcal{G}_2^\top = \begin{pmatrix} \chi_{1}(\lambda)^3 \end{pmatrix} = \begin{pmatrix} (1)^3 \end{pmatrix} = \begin{pmatrix} 1 \end{pmatrix}. $$
		Therefore $\log_{\xi}(\mathcal{G}_1 \odot_M \mathcal{G}_2^\top) = \begin{pmatrix} 0 \end{pmatrix}$. Since $k = 1$,  $\ell = k - \mathrm{rk} = 1 - 0 = 1$ (by Theorem~\ref{th-3.1}), the pair $\{\mathcal{C}_1, \mathcal{C}_2\}$ forms a $1$-RAIP.
		
		Now, choose $\mathbf{a}_1 = (1, 1, 0)$ and $\mathbf{a}_2 = (\lambda, 0, \lambda)$ and we construct the expanded $(3, 2^2)$ codes $\mathcal{C}_1'$ and $\mathcal{C}_2'$ with generator matrices
		\[
		\mathcal{G}_1' = \begin{pmatrix} 1 & 1 & 0 \\ 1 & 1 & 1 \end{pmatrix}, \quad \mathcal{G}_2' = \begin{pmatrix} \lambda & 0 & \lambda \\ \lambda & \lambda & \lambda \end{pmatrix}.
		\]
		Now,
		$$ \mathcal{G}_1' \odot_M (\mathcal{G}_2')^\top = \begin{pmatrix} 1 & 1 \\ 1 & 1 \end{pmatrix}. $$
		Thus,
		$$ \log_{\xi}\left(\mathcal{G}_1' \odot_M (\mathcal{G}_2')^\top\right) = \begin{pmatrix} 0 & 0 \\ 0 & 0 \end{pmatrix}. $$
		Since the new matrix dimension is $k' = 2$, by Theorem~\ref{th-3.1}, $\ell' = k' - \mathrm{rk} = 2 - 0 = 2$. Therefore, the pair $\{\mathcal{C}_1', (\mathcal{C}_2')^M\}$ forms a $2$-RAIP of codes, verifying Theorem~\ref{th-4.4}.
\end{example}}

\section{$\ell$-RAIP of codes from GRS and extended GRS codes}\label{sec:5}

The main aim of this section is to focus on the construction of $\ell$-RAIP of codes from GRS and extended GRS codes. Let $\ba=(a_1,a_2,\cdots,a_{mn})$ and $\mathbf{v}=(v_1, v_2, \ldots, v_{mn})$ be two vectors in $\FF_q^{mn}$ such that $mn\leq q$, with $a_1, a_2, \ldots a_{mn}$ being a distinct element of $\FF_q$ and $v_i$ is a non-zero element of $\FF_q$, for $i=1,2,\ldots, mn$. Also, let $P(x)$ be a polynomial in $\FF_q$ of degree less than or equal to $nm$ such that $P(a_i)\neq 0$, for $i=1,2,\ldots, mn$.
According to Jin's paper \cite{JIN}, we define generalized Reed Solomon (GRS) code as 
$$\textit{GRS}{(\ba, P(x), \mathbf{v})}=\Bigl\{ \left(\dfrac{v_1f(a_1)}{P(a_1)}, \dfrac{v_2f(a_2)}{P(a_2)}, \ldots, \dfrac{v_1f(a_{mn})}{P(a_{mn})}\right)\mid f(x)\in\FF_q;~\deg f(x)< \deg P(x)\Bigl\}.$$ Note that, the parameter of this code is $[nm, \deg P(x), nm-\deg P(x)+1]$.

The following result can be proved using a similar methodology to that used to prove Theorem~3.1 in \cite{Guenda2019}.

\begin{prop}\cite[Theorem~3.1]{Guenda2019}\label{p-4.1}
	Let $\textit{GRS}{(\ba, P(x), \mathbf{v})}$  and $\textit{GRS}{(\ba, Q(x), \mathbf{v})}$ be two GRS codes such that  and satisfying the following conditions
	\begin{enumerate}
		\item[1)] $\deg P(x)+ \deg Q(x)=nm$.
		\item[2)] $\gcd(P(x)Q(x),\prod\limits_{i=1}^{mn}(x-a_i))=1$.
	\end{enumerate}
	Then $\textit{GRS}{(\ba, P(x), \mathbf{v})} \cap \textit{GRS}{(\ba, Q(x), \mathbf{v})}=\textit{GRS}{(\ba, \gcd(P(x), Q(x)), \mathbf{v})}$.
\end{prop}
\begin{thm}
	Let $\textit{GRS}{(\ba, P(x), \mathbf{v})}$  and $\textit{GRS}{(\ba, Q(x), \mathbf{v})}$ be two GRS codes such that  and satisfying the following conditions
	\begin{enumerate}
		\item[1)] $\deg P(x)+ \deg Q(x)=nm$.
		\item[2)] $\gcd(P(x)Q(x),\prod\limits_{i=1}^{mn}(x-a_i))=1$.
	\end{enumerate}
	Then the pair $\{\Phi^{-1}\textit{GRS}{(\ba, P(x), \mathbf{v})}, \Phi^{-1}\textit{GRS}{(\ba, Q(x), \mathbf{v})}\}$ is a $\ell$-RAIP of codes over $\FF_q^{m}$, where $\ell=\deg\gcd(P(x), Q(x))$. 
\end{thm}
\begin{proof}
	
	Since \(\Phi\) is an isomorphism, the sets \(\{\Phi^{-1}\text{GRS}(\mathbf{a}, P(x), \mathbf{v})\}\) and \(\{\Phi^{-1}\text{GRS}(\mathbf{a}, Q(x), \mathbf{v})\}\) are two additive codes of length \(n\) over \(\mathbb{F}_{q^m}\). By applying Lemma~\ref{lm-3.1} and Remark~\ref{rk-1}, we find that \(\text{rk}_q\left(\{\Phi^{-1}\text{GRS}(\mathbf{a}, P(x), \mathbf{v})\} \cap \{\Phi^{-1}\text{GRS}(\mathbf{a}, Q(x), \mathbf{v})\}\right) = \dim\left(\text{GRS}(\mathbf{a}, P(x), \mathbf{v}) \cap \text{GRS}(\mathbf{a}, Q(x), \mathbf{v})\right)\). Thus, from Proposition~\ref{p-4.1}, we can immediately conclude the desired result.\end{proof}

Next, we provide the construction of \(\ell\)-RAIP codes from the extended GRS code. First, let us examine the structure of the extended GRS code. Suppose that \(r(x) = \frac{a(x)}{b(x)} = \frac{\sum\limits_{i=0}^t a_{i} x^i}{\sum\limits_{i=0}^t b_{i} x^i}\). If \(b_t \neq 0\), then \(r(\infty) = \frac{a_t}{b_t}\). From this, we conclude that \(r(\infty) = 0\) if and only if \(\deg r(x) < 0\), where $\deg r(x) = \deg a(x) - \deg b(x)$.

According to Jin's paper \cite{JIN}, for \(\mathbf{b} = (b_1, b_2, \ldots, b_{mn})\) with \(b_i \in \mathbb{F}_q^*\) and \(nm \leq q+1\), we define the extended GRS code as follows

$$\textit{GRS}_{\infty}{(\bb, P(x), \mathbf{v})}=\Bigl\{ \left(\dfrac{v_1f(b_1)}{P(b_1)}, \dfrac{v_2f(b_2)}{P(b_2)}, \ldots, \dfrac{v_1f(b_{mn})}{P(b_{mn})}, v_n\left(\dfrac{xf}{P}\right)\right)\mid f(x)\in\FF_q;~\deg f(x)< \deg P(x)\Bigl\}.$$ The obtained  code is precisely an $[nm, \deg P(x), nm-\deg P(x)+1]$-code.

We can show the following result by following a proof similar to that of Theorem~3.2 in \cite{Guenda2019}.
\begin{prop}\cite[Theorem~3.2]{Guenda2019}\label{p-4.2}
	Let $\textit{GRS}_{\infty}{(\bb, P(x), \mathbf{v})}$  and $\textit{GRS}_{\infty}{(\bb, Q(x), \mathbf{v})}$ be two extended GRS codes such that  and satisfying the following conditions
	\begin{enumerate}
		\item[1)] $\deg P(x)+ \deg Q(x)=nm$.
		\item[2)] $\gcd(P(x)Q(x),\prod\limits_{i=1}^{mn}(x-a_i))=1$.
	\end{enumerate}
	Then $\textit{GRS}_{\infty}{(\bb, P(x), \mathbf{v})} \cap \textit{GRS}_{\infty}{(\bb, Q(x), \mathbf{v})}=\textit{GRS}_{\infty}{(\bb, \gcd(P(x), Q(x)), \mathbf{v})}$.
\end{prop}

\begin{thm}
	Let $\textit{GRS}_{\infty}{(\bb, P(x), \mathbf{v})}$  and $\textit{GRS}_{\infty}{(\bb, Q(x), \mathbf{v})}$ be two GRS codes  satisfying the following conditions
	\begin{enumerate}
		\item[1)] $\deg P(x)+ \deg Q(x)=nm$.
		\item[2)] $\gcd(P(x)Q(x),\prod\limits_{i=1}^{mn}(x-a_i))=1$.
	\end{enumerate}
	Then,  the pair $\{\Phi^{-1}\textit{GRS}_{\infty}{(\bb, P(x), \mathbf{v})}, \Phi^{-1}\textit{GRS}_{\infty}{(\bb, Q(x), \mathbf{v})}\}$ is a $\ell$-RAIP of codes over $\FF_q^{m}$, where $\ell=\deg\gcd(P(x), Q(x))$. 
\end{thm}
\begin{proof}
	Since, $\Phi$ is an isomorphism, then  $\{\Phi^{-1}\textit{GRS}_{\infty}{(\bb, P(x), \mathbf{v})}$ and $\Phi^{-1}\textit{GRS}_{\infty}{(\bb, Q(x), \mathbf{v})}\}$ are two additive codes of length $n$ over $\FF_{q^m}$. By applying Lemma~\ref{lm-3.1} and Remark~\ref{rk-1}, we have $$\mathrm{rk}_q\left(\{\Phi^{-1}\textit{GRS}_{\infty}{(\bb, P(x), \mathbf{v})}\cap \Phi^{-1}\textit{GRS}_{\infty}{(\bb, Q(x), \mathbf{v})}\right)=\dim\{\textit{GRS}_{\infty}{(\bb, P(x), \mathbf{v})}\cap \textit{GRS}{(\bb, Q(x), \mathbf{v})}\}.$$ Thus, from Proposition~\ref{p-4.2}, the desired result follows immediately.  
\end{proof}

\section{Conclusion}\label{sec:conclusion}

We have introduced \(\ell\)-rank additive intersection pairs (RAIP) of codes over finite fields and examined their properties (namely related to their duality) through their generator and parity-check matrices. Our study investigated the relationships between these intersection pairs and several well-known families of codes, including additive complementary (ACD) and additive complementary (ACP) codes, as well as the hull of an additive code. We explored their interconnections, categorized them, derived several constructions, and contributed to their classification. This type of code pair builds on various concepts in the literature, specifically ACD and ACP codes, as well as the hull of an additive code. As a result, future research avenues appear promising. These could lead to significant applications, particularly for complementary pairs of codes related to side-channel and fault injection attacks, which have recently emerged as viable countermeasures against passive and active side-channel analyses in embedded cryptosystems. An interesting direction for future work is the investigation of $\ell$-RAIP codes with respect to Hermitian duality. Such a study may provide new insights into the design of quantum error-correcting codes that allow varying amounts of entanglement, following the framework introduced by Luo et al.~\cite{LFGS24}.

\section*{Acknowledgement}  The authors sincerely thank the Associate Editor and the referees for their valuable comments and constructive suggestions, which have significantly improved the quality and presentation of this paper. The first author would like to thank the Indian Institute of Technology Guwahati for its hospitality and support.


\begin{thebibliography}{00}
	
	\bibitem{Agrawal24} 
	A. Agrawal, R. K. Sharma.
	\newblock ACD codes over skew-symmetric dualities.
	\newblock{\em Cryptogr. Commun.}, \textbf{14}:1013–1032, 2024.
	
	\bibitem{Agrawa24} 
	A. Agrawal, R. K. Sharma.
	\newblock Additive one-rank hull codes over finite fields.
	\newblock{\em Finite Fileds Appl.}, \textbf{96}:102426, 2024.
	
	\bibitem{ACLMRS24}
	S.~E. Anderson, E.~Camps-Moreno, H.~H. L{\'o}pez, G.~L. Matthews, D.~Ruano, and I.~Soprunov.
	\newblock Relative hulls and quantum codes.
	\newblock {\em IEEE Trans. Inf. Theory}, \textbf{70}(5):3190--3201, 2024.
	
	\bibitem{BDM23}
	S. Bhowmick, D. K. Dalai, and S. Mesnager.
	\newblock On Linear Complementary Pairs of Algebraic Geometry Codes over Finite Fields.
	\newblock In {\em  Discrete Mathematics}, \textbf{347}(12):114193, 2024.
	
	\bibitem{BD24}
	S. Bhowmick and D. K. Dalai.
	\newblock Additive complementary pairs of codes.
	\newblock{\em Adv. Math. Commun.}, \textbf{19}(6):1694-1712, 2025.
	
	
	\bibitem{Bhowmick23} S. Bhowmick, A. F. Tabue, J. Pal.
	\newblock On the $\ell$-DLIPs of codes over finite commutative rings.
	\newblock{Discrete Mathematics,} \textbf{347}(4):113853, 2024.
	
	
	\bibitem{Bierbrauer}
	J. Bierbrauer, Y. Edel, G. Faina, S. Marcugini, F. Pambianco.
	\newblock Short additive quaternary codes.
	\newblock {\em IEEE Trans. Inform. Theory}, \textbf{55}:952-954, 2009.
	
	\bibitem{Blokhuis04} A. Blokhuis, A. Brouwer.
	\newblock Small additive quaternary codes.
	\newblock{Eur. J. Comb.} \textbf{25}(2):161-167, 2004.
	
	\bibitem{Cal98} 
	A.R. Calderbank, E.M. Rains, P.M. Shor, N.J.A. Sloane.
	\newblock Quantum error correction via codes over $GF(4)$.
	\newblock{\em Phys. Rev. A}, \textbf{44}(2):1369–1387, 1998. 
	
	\bibitem{CG18}
	C. Carlet, C. G{\"{u}}neri, F. {\"{O}}zbudak, B. {\"{O}}zkaya, and P. Sol{\'{e}}.
	\newblock On linear complementary pairs of codes.
	\newblock {\em }, \textbf{64}(10):6583--6589, 2018.
	
	\bibitem{CMT18}
	C. Carlet, S. Mesnager, C. Tang, Y. Qi, and R. Pellikaan.
	\newblock Linear codes over $\FF_q$ are equivalent to {LCD} codes for $q>3$.
	\newblock {\em IEEE Trans. Inform. Theory}, \textbf{64}(4):3010--3017, 2018.
	
	\bibitem{CMTQ19}
	C. Carlet, S. Mesnager, C. Tang, and Y. Qi.
	\newblock On $\sigma$-LCD codes.
	\newblock {\em IEEE Trans. Inform. Theory}, \textbf{65}(3):1694--1704, 2019.
	
	\bibitem{choi23} 
	W. H.  Choi, C. G{\"{u}}neri, J. L. Kim, F. {\"{O}}zbudak.
	\newblock Theory of additive complementary dual codes, constructions and computations.
	\newblock{\em Finite Fileds Appl.}, \textbf{92}:102303, 2023.
	
	\bibitem{Guenda2019}  K. Guenda, T. A. Gulliver, S. Jitman, S. Thipworawimon.
	\newblock Linear $\ell$-Intersection Pairs of Codes and Their Applications. 
	\newblock{\em Designs, Codes and Cryptography,} \textbf{88}(1), 133-152, 2020.
	
	\bibitem{Del73} P. Delsarte.  
	\newblock An algebraic approach to the association schemes of coding theory. 
	\newblock{\em Philips Res.Rep. Suppl.}, 1973.
	
	\bibitem{Dougherty22} S.T. Dougherty, S. Şahinkaya, D. Ustun.  
	\newblock Additive complementary dual codes from group characters. 
	\newblock{\em IEEE Trans. Inform. Theory}, \textbf{68}(7): 4444-4452, 2022.
	
	\bibitem{Dou124} S.T. Dougherty, S. Şahinkaya, D. Ustun.
	\newblock On additive codes with one-rank hulls. 
	\newblock{\em Appl. Algebra Eng. Commun. Comput.}, https://doi.org/10.1007/s00200-024-00663-5.
	
	\bibitem{Ezerman} M. F. Ezerman, S. Ling and P. Sol\'e.  
	\newblock Additive asymmetric quantum codes. 
	\newblock{\em IEEE Trans. Inform. Theory}, \textbf{57}: 5536--5550, 2011.
	
	\bibitem{Huff13} W. C. Huffman.  
	\newblock On the theory of $\FF_q$-linear $\FF_{q^t}$-codes. 
	\newblock{\em Advances in Mathematics of Communications}, \textbf{7}: 349--378, 2013.
	
	\bibitem{Kim17} J. L. Kim and N. Lee.  
	\newblock Secret sharing schemes based on additive codes over GF(4). 
	\newblock{\em Appl. Algebra Eng. Commun. Comput.}, \textbf{28}: 79--97, 2017.
	
	\bibitem{LFGS24} 
	G. Luo, M.-F. Ezerman, M. Grassl, and S. Ling.
	\newblock Constructing quantum error-correcting codes that require a variable amount of entanglement.
	\newblock{\em Quantum Information Processing}, \textbf{23}(4), 2024.
	
	\bibitem{Mas92}
	J. L. Massey.
	\newblock Linear codes with complementary duals.
	\newblock {\em Discrete Mathematics}, \textbf{106-107}:337--342, 1992.
	
	\bibitem{JIN} L. Jin.  
	\newblock Construction of MDS codes with complementary duals. 
	\newblock{\em IEEE Transactions on Information Theory}, \textbf{63}: 2843–2847, 2017.
	
	\bibitem{shi22} 
	M. Shi, S. Liu, J. L. Kim, and P. Sol{\'{e}}.
	\newblock Additive complementary dual codes over $\FF_4$.
	\newblock{\em Des., Codes Cryptogr.}, \textbf{91}:273-284, 2022.
	
	
	
	\bibitem{shi222} 
	M. Shi, S. Liu, F. {\"{O}}zbudak, and P. Sol{\'{e}}.
	\newblock Additive cyclic complementary dual codes over $\FF_4$.
	\newblock{\em Finite Fileds Appl.}, \textbf{75}:102087, 2022.
	
	\bibitem{Woo99} J. Wood. 
	\newblock Duality for modules over finite rings and applications to coding theory. 
	\newblock{ American Journal of Mathematics,} \textbf{121}, 555--575, (1999).
	
	
	
\end{thebibliography}
\end{document}